\documentclass[aps,pra,twocolumn,superscriptaddress,showpacs,final,nofootinbib,floatfix]{revtex4}

\usepackage{amsmath,amsfonts,amsthm,amssymb}
\usepackage{graphicx}
\usepackage{braket}
\usepackage{setspace}
\usepackage{hyperref}
\usepackage{bbm}

\usepackage{tikz}
\usepackage{tikzscale}
\usepackage{standalone}
\usetikzlibrary{arrows}
\usetikzlibrary{positioning}
\usetikzlibrary{patterns}

\newtheorem{proposition}{Proposition}
\newtheorem{proposition?}{Proposition?}

\theoremstyle{definition}

\newcommand{\real}{\mathbb R} 
\newcommand{\half}{\frac{1}{2}} 
\newcommand{\mo}[1]{\left| #1 \right|} 

\newcommand{\ip}[2]{\left\langle\,#1\,|\,#2\,\right\rangle} 
\newcommand{\kb}[2]{|#1\rangle\langle #2|} 
\newcommand{\tr}[1]{\mathrm{tr}\left[#1\right]} 
\newcommand{\ptra}[1]{\mathrm{tr}_A[#1]} 

\newcommand{\id}{\mathbbm{1}} 

\newcommand{\va}{\mathbf{a}} 
\newcommand{\vb}{\mathbf{b}} 
\newcommand{\vc}{\mathbf{c}} 
\newcommand{\vx}{\mathbf{x}} 
\newcommand{\vy}{\mathbf{y}} 
\newcommand{\vz}{\mathbf{z}} 
\newcommand{\vsigma}{\boldsymbol{\sigma}} 

\newcommand{\A}{\mathsf{A}}
\newcommand{\B}{\mathsf{B}}
\newcommand{\Go}{\mathsf{G}}

\newcommand{\So}{\mathsf{S}}

\begin{document}

\title{Adaptive strategy for joint measurements}

\author{Roope Uola}
\affiliation{Naturwissenschaftlich-Technische Fakult\"at, 
Universit\"at Siegen, Walter-Flex-Str.~3, 57068 Siegen, Germany}

\author{Kimmo Luoma}
\affiliation{Institut f{\"u}r Theoretische Physik, Technische Universit{\"a}t Dresden, 
D-01062,Dresden, Germany}

\author{Tobias Moroder}
\affiliation{Naturwissenschaftlich-Technische Fakult\"at, 
Universit\"at Siegen, Walter-Flex-Str.~3, 57068 Siegen, Germany}

\author{Teiko Heinosaari}
\affiliation{Turku Center for Quantum Physics, Department of Physics and 
Astronomy, University of Turku, Finland}

\begin{abstract}
We develop a technique to find simultaneous measurements for noisy quantum observables in finite-dimensional Hilbert spaces. We use the method to derive lower bounds for the noise needed to make incompatible measurements jointly measurable. Using our strategy together with recent developments in the field of one-sided quantum information processing we show that the attained lower bounds are tight for various symmetric sets of quantum measurements. We use this characterisation to prove the existence of so called 4-Specker sets in the qubit case.
\end{abstract}


\maketitle

\section{Introduction}

Incompatibility of quantum devices is one of the fundamental features of quantum theory with a wide range of consequences \cite{HeMiZi16}. In particular, incompatibility of quantum measurements is known to be a very powerful tool in many branches of quantum information theory including entanglement detection \cite{Gühne04}, uncertainty relations \cite{Busch14, Busch13} and tasks demanding a Bell violation like quantum cryptography \cite{WoPeFe09, Brunner14}. Recent developments \cite{Heinosaari15, MQ, RU1} are suggesting that incompatibility can be seen as a resource for quantum information processing.
The resource theoretical aspect has spurred a development of monotones quantifying quantum incompatibility or non-joint measurability \cite{Heinosaari152, Pusey15, RU2} as a general quantum resource. These monotones are based on adding noise to a set of incompatible measurements and concluding numerically the noise threshold for the measurements to become jointly measurable. Besides the numerical methods a few measurement setups have been analysed analytically, including two \cite{Busch86, BuHe08} and three qubit measurements \cite{Yu13, LiSpWi11}, two mutually unbiased measurements \cite{CaHeTo12, Haapasalo15}, and a measurement setup including a Clifford algebra generalisation of the Pauli matrices \cite{KuHeFr14}.
There are also some general methods to derive either necessary \cite{Heinosaari13} or sufficient \cite{Zhu15} conditions for a set of measurements to be incompatible.

In this article we present a method to derive analytical bounds for the noise resistance of incompatible measurements on a $d$-level quantum system. Our method is based on an adaptive algorithm which starts with a set of well chosen Hilbert space vectors and results as a set of noisy compatible measurements. The algorithm gives lower bounds for the noise needed to make a set of $M$ quantum measurements compatible. 
We demonstrate the power of our method by re-deriving some of the known \cite{Busch86, CaHeTo12} joint measurement uncertainty relations, but then also applying the technique to various new symmetric measurement scenarios obtaining lower bounds for the noise robustness of the involved measurements. Moreover, we translate known quantum steering \cite{JoWi11, Saunders10} techniques into inequalities which, when violated, witness incompatibility.  These inequalities are shown to coincide with our lower bounds proving the optimality of our results. We conclude by constructing a so-called 4-Specker set of incompatible quantum measurements in the qubit scenario.

The paper is organized as follows: in Section II we recall the mathematical definition of joint measurability and we explain the general idea of the method, in Section III we exploit the strategy for qubit measurements, in Section IV we generalize our strategy for $d$-level quantum systems and apply it to two mutually unbiased bases (MUBs), and in Section V we prove the optimality of our results using known steering techniques. We conclude by proving the existence of a 4-Specker set in the qubit case in section VI and by stating our conclusions in section VII.

\section{Quantum incompatibility and the adaptive strategy}
 
Quantum incompatibility means the impossibility of measuring two or more quantum observables simultaneously. For the case of projective 
measurements incompatibility is characterized 
by the non-commutativity of the measurements but for general observables (i.e. positive operator valued measures or POVMs for short) this is no longer the case. Indeed, there exists 
non-commuting observables which allow a simultaneous measurement. 
Generally, incompatibility is formulated as 
the non-existence of a joint measurement: a set $\{\A_k\}_{k=1}^M$ of observables is compatible if there exist a joint observable $\Go$ from which one recovers the observables as marginals, i.e.
\begin{equation}\label{JMDef}
\A_k(x_k)=\sum_{x_i, i\neq k} \Go(x_1,...,x_m).
\end{equation}
A set of observables which is not compatible is called incompatible. In what follows, we develop an intuitive strategy for implementing joint measurements for several quantum measurement scenarios.

Suppose we have measurement devices for two observables $\A_1$ and $\A_2$, but that these observables are incompatible.
We can thus only hope to simultaneously implement their \emph{approximate} versions.
But how to perform a joint measurement even of their approximate versions if the only available devices are the ones for $\A_1$ and $\A_2$? 

A rough method as presented in Fig.~\ref{fig:joint_guess} is to toss a coin and measure either $\A_1$ or $\A_2$, depending on the result of the coin tossing.
Since a joint measurement should give an outcome for both observables, one can then for example draw randomly an outcome for the other observable.
A POVM describing this procedure is
\begin{align}\label{eq:1}
\Go(\alpha_1,\alpha_2) = \mu p(\alpha_2) \A_1(\alpha_1) + (1-\mu) p(\alpha_1) \A_2(\alpha_2) \, ,
\end{align}
where $p$ is a probability distribution and $\mu$ is the probability that the coin tossing leads to the choice of the first observable.
The noisy versions of $\A_1$ and $\A_2$ are $\mu \A_1 + (1-\mu) p \id$ and $(1-\mu) \A_2 + \mu p \id$, respectively.
 
\begin{figure}
  \includegraphics{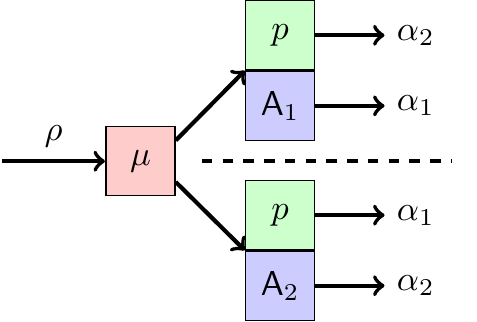}
  \caption{\label{fig:joint_guess} (Color online) Non-adaptive strategy to obtain a joint observable for approximations of two observables. A system in a state $\rho$ enters the measurement device. With probability $\mu$ the observable $\A_1$ is measured and for observable $\A_2$ a random outcome is drawn. With probability $1-\mu$ the observable $\A_2$ is measured and a value for the observable $\A_1$ is drawn randomly.}
\end{figure}

There is a way to improve the previous procedure if we know something about the relation of $\A_1$ and $\A_2$. 
Namely, we can take an advantage of the obtained measurement outcome and adapt the random choice accordingly. 
This means that we replace the probability distribution $p$ by conditional probability distributions, hence leading to a joint observable
\begin{align}\label{eq:2}
\Go(\alpha_1,\alpha_2) =& \mu p(\alpha_2 \mid \A_1 = \alpha_1 ) \A_1(\alpha_1) \notag\\ 
                        &+ (1-\mu) p(\alpha_1 \mid \A_2 = \alpha_2 ) \A_2(\alpha_2) \, .
\end{align}
This has an obvious generalization to any finite number of observables. The strategy is illustrated in Fig.~\ref{fig:joint_adaptive}.

\begin{figure}
  \includegraphics{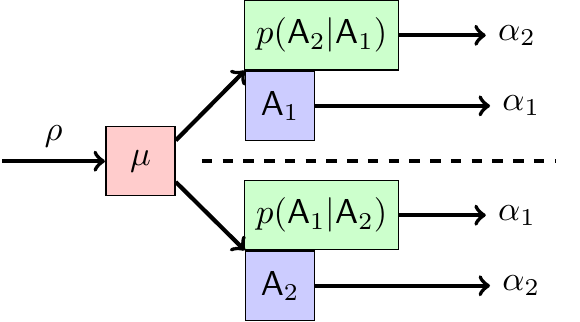}
  \caption{\label{fig:joint_adaptive} (Color online) Adaptive strategy to obtain a joint observable for approximations of two observables.. In comparison to the non-adaptive strategy (Fig. \ref{fig:joint_guess}) in the adaptive strategy the outcome of the non-measured observable is conditioned on the outcome of the measured observable.}
\end{figure}

There is still an important modification of the previous procedure that adds a new dimension of flexibility (see Fig.~\ref{fig:joint_adaptive_aux}). Suppose we want to perform a joint measurement of noisy versions of $\A_1,\ldots,\A_M$, but we have measurement devices for some 
other (incompatible) observables $\B_1,\ldots,\B_N$ in our possessions. 
Again, we make a random choice of which observable $\B_k$ we measure. 
Based on the obtained outcome, we create $M$ outcomes $\alpha_1,\ldots,\alpha_M$ that are interpreted as the outcomes of noisy versions 
of $\A_1,\ldots,\A_M$.
If an observable $\B_k$ gives an outcome $\beta_k$, then the resulting set of outcomes $\alpha_1,\ldots,\alpha_M$ is obtained with some 
conditional probability $p(\alpha_1, \ldots,\alpha_M \mid \B_k = \beta_k )$. 
We also note that the observables $\B_k$ do not have to be chosen with equal probability, but in general we can throw an $N$-sided 
biased dice giving the outcome $k$ (i.e. telling to measure $\B_k$) with probability $\mu_k$.
In this general case the obtained POVM is thus
\begin{align}\label{eq:3}
\Go(\alpha_1,\ldots,\alpha_M) =& \mu_1 \sum_\beta p(\alpha_1,\ldots,\alpha_M \mid \B_1 = \beta_1 ) \B_1(\beta_1) 
                               + \cdots\notag\\
                               + &\mu_N \sum_\beta p(\alpha_1, \ldots,\alpha_M \mid \B_N = \beta_N ) \B_N(\beta_N).
\end{align}

\begin{figure}
  \includegraphics{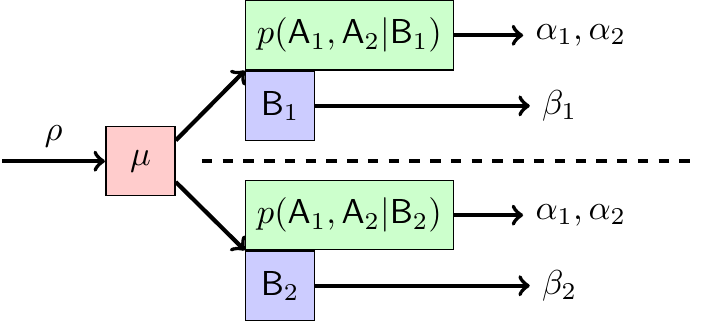}
  \caption{\label{fig:joint_adaptive_aux} (Color online)
    Adaptive strategy with auxiliary observables. In the general case the adaptive strategy uses some other observables $\B_1, \B_2$ to obtain approximations of the original observables $\A_1,\A_2$.}
\end{figure}

There is no guarantee that the marginals of this joint observable are good approximations of $\A_1,\ldots,\A_M$. This obviously depend on the choices of $\B_1,\ldots,\B_N$ and the conditional post processing. In what follows we show that our technique captures many of the known examples of joint measurability and, moreover, we present various new joint measurement scenarios.

\section{Adaptive strategy for qubit observables}\label{Adaptivesec}

The general procedure described in the last sections has two critical choices: test observables $\B_1,\ldots,\B_N$ and conditional post-processing functions  $p(\alpha_1, \ldots,\alpha_M \mid \B_k= \beta )$.
We will now specify the latter in the case of unbiased qubit observables.

First, we start from unbiased binary qubit observables.
For each unit vector $\va\in\real^3$ and $0\leq \lambda \leq 1$, we denote by $\So^{\lambda\va}$ the binary qubit observable
\begin{align*}
\So^{\lambda \va}(\pm 1) = \half ( \id \pm \lambda\va\cdot\vsigma ) \, .
\end{align*}
For values $0<\lambda<1$ we consider $\So^{\lambda \va}$ as a noisy version of the sharp qubit observable $\So^{\va}$ \cite{Busch86, BuHe08}.

Let $\va_1,\ldots,\va_M\in\real^3$ be a finite set of unit vectors.
We are seeking for a joint observable for noisy versions of $\So^{\va_1},\ldots,\So^{\va_M}$, i.e., we want to construct an observable $\Go$ such that
\begin{align*}
\sum_{\alpha_2,\ldots,\alpha_M=\pm 1} \Go(\alpha_1,\ldots,\alpha_M) = \So^{\lambda\va_1}(\alpha_1)\, ,
\end{align*}
and similarly for the other marginals. 
Following the general guideline of the adaptive strategy, we proceed as follows.

\begin{enumerate}
\item We fix a set of unit vectors $\vb_1,\ldots,\vb_N\in\real^3$ such that $\va_\ell \cdot \vb_k \neq 0$ for 
  all $\ell =\{ 1,\ldots,M\}$ and $k =\{ 1,\ldots,N\}$.
\item We choose randomly $k\in \{1,\ldots,N\}$.
\item We perform a measurement of $\So^{\vb_k}$ in the input state $\varrho$, hence obtain an outcome $\beta_k=\pm 1$ with the 
  probability $\tr{\varrho \So^{\vb_k}(\pm 1)}$.
\item For each $\ell =\{ 1,\ldots,M\}$, we decide that the outcome $\alpha_\ell$ is $\beta_k$ if $\va_\ell \cdot \vb_k > 0$ 
  and $-\beta_k$ if $\va_\ell \cdot \vb_k < 0$. 
\item As a result, we get a list $(\alpha_1,\ldots,\alpha_M)$ of outcomes. 
\item It is possible that some combination $(\alpha_1,\ldots,\alpha_M)$ does not result in the process at all. 
  In this case we set $\Go(\alpha_1,\ldots,\alpha_M)=0$.
\end{enumerate}

The marginals of $\Go$ are binary qubit observables, but they are not guaranteed to be unbiased noisy versions of the original 
observables $\So^{\va_1},\ldots,\So^{\va_M}$.
However, we will next see that the marginals are unbiased noisy versions of the original observables in many symmetric situations if the vectors 
$\vb_1,\ldots,\vb_N$ are chosen properly.

\subsection{Planar directions}\label{Planarsec}

We begin our discussion with a well-known example of two orthogonal qubit observables $\So^{\bf x}$ and $\So^{\bf y}$ \cite{Busch86}.
 We want to build our joint measurement by finding two measurement directions 
and then randomly performing one of the measurements. For this purpose we employ projective measurements in the directions which are 
equal superpositions of the Bloch vectors ${\bf x}, {\bf y}$ of $\So^{\bf x}(+)$ and $\So^{\bf y}(+)$. If an outcome of a measurement in the direction $\vb_1 = \frac{1}{\sqrt 2}({\bf x}+{\bf y})$ is positive 
(resp. negative), then the measurement outcomes of $\So^{\bf x}$ and $\So^{\bf y}$ are decided to be positive 
(resp. negative) as $\vx\cdot\vb_1>0$ and $\vy\cdot\vb_1>0$. In a similar way, if a measurement in the direction $\vb_2 = \frac{1}{\sqrt 2}({\bf x}-{\bf y})$ is performed and a positive (resp. negative) outcome is obtained, then the measurement outcomes of $\So^{\bf x}$ and $\So^{\bf y}$ are decided to be $+$ and $-$ (resp. $-$ and $+$) as $\vx\cdot\vb_2>0$ and $\vy\cdot\vb_2<0$. In this procedure the actually measured observable 
is a POVM $\Go$ given by the operators
\begin{align*}
\Go(+,+)&=\frac{1}{2}\So^{\vb_1}(+) \, ,& \Go(-,-) &=\frac{1}{2}\So^{\vb_1}(-)\, ,\nonumber\\
\Go(+,-)&=\frac{1}{2}\So^{\vb_2}(+)\, ,
& \Go(-,+)& =\frac{1}{2}\So^{\vb_2}(-).
\end{align*}
The marginals (post-processings) of $\Go$ are given as
\begin{align}
\Go(+,+)+\Go(+,-)&=\lambda \So^{\bf x}(+)+(1-\lambda)\frac{1}{2}\id,\\
\Go(+,+)+\Go(-,+)&=\lambda \So^{\bf y}(+)+(1-\lambda)\frac{1}{2}\id,
\end{align}
with $\lambda=\frac{1}{\sqrt 2}\approx 0.7071$. The parameter $\lambda$ is usually called the noise parameter. It tells how much white noise is added to the observable. Surprisingly, the obtained value of $\lambda$ is known to be necessary and sufficient for the 
joint measurability of the observables $\So^{\bf x}$ and $\So^{\bf y}$ \cite{Busch86}.
\begin{figure}
  \includegraphics[width=0.49\textwidth]{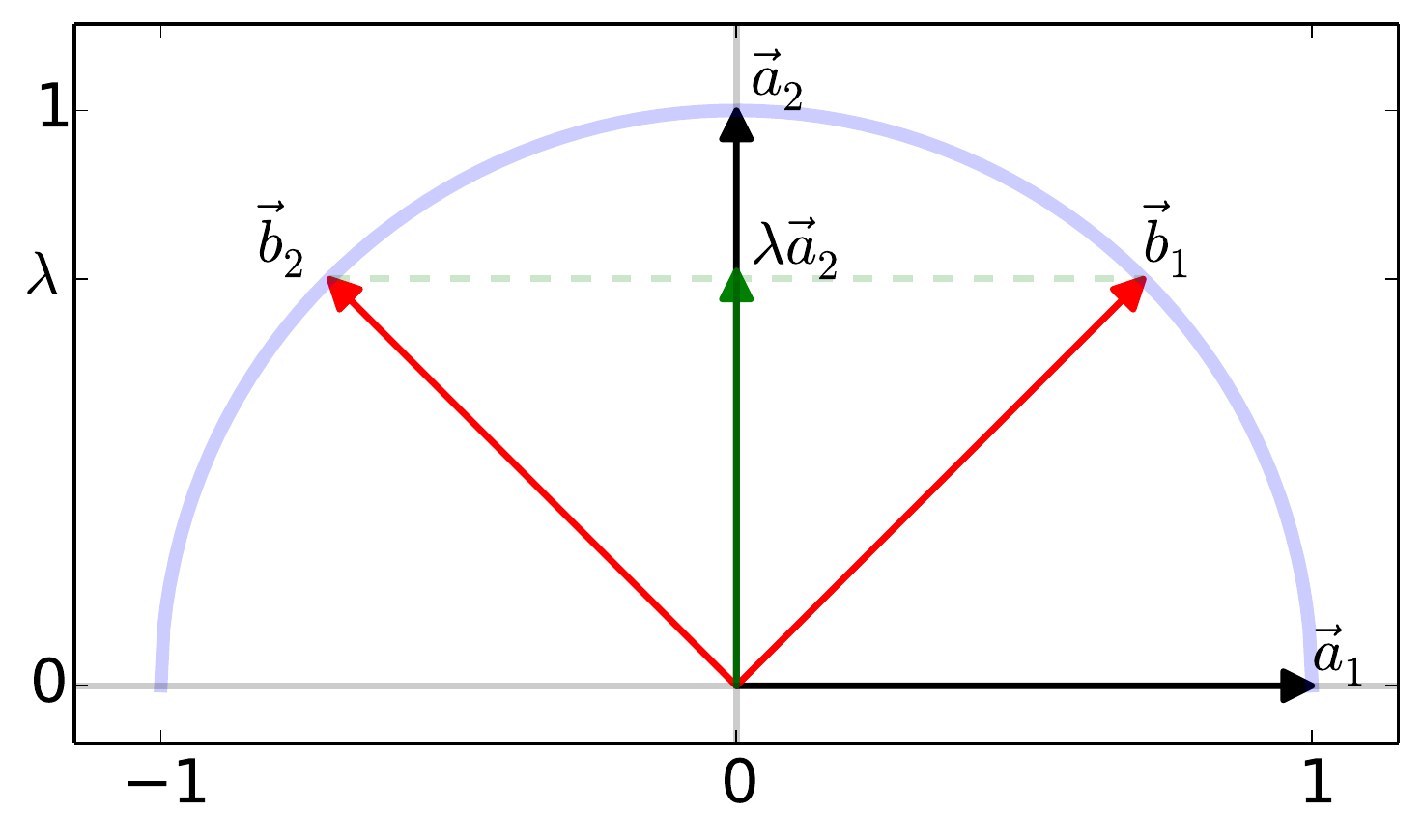}
  \caption{\label{fig:TwoObservablesInPlane} (Color online) Two orthogonal observables in plane.}
\end{figure}

The following proposition generalizes the previous example for symmetrical arrangement of qubit observables in a plane.

\begin{proposition}\label{sec:observables-plane}
Let $M\geq 2$ be an integer and 
\begin{equation*}
\va_k = \cos \theta_k \vx + \sin\theta_k \vy \, , \quad \theta_k=(k-1) \pi/M
\end{equation*}
 for $k=1,\ldots,M$.
The observables $\So^{\lambda\va_1},\ldots,\So^{\lambda\va_M}$ are jointly measurable if
\begin{align*}
\lambda \leq  \frac{1}{M \sin(\frac{\pi}{2M})} \, .
\end{align*}
\end{proposition}

\begin{proof}
\emph{Suppose that $M$ is odd}.
We choose $\vb_k=\va_k$ for $k=1,\ldots,M$ and follow the previously described procedure.
We have
\begin{align*}
\va_k \cdot \va_\ell & = \cos \theta_k \cos \theta_\ell + \sin\theta_k \sin\theta_\ell \\
& = \cos (\theta_k - \theta_\ell) \\
& = \cos \frac{(k-\ell)\pi}{M} \, .
\end{align*}
Hence,
\begin{align*}
\va_k \cdot \va_\ell > 0 \quad \mathrm{if} \ \mo{k-\ell}<M/2, \\
\va_k \cdot \va_\ell < 0 \quad \mathrm{if} \ \mo{k-\ell}>M/2.
\end{align*}
The first marginal of $\Go$ is
\begin{align*}
\sum_{\alpha_2,\ldots,\alpha_M=\pm 1} \Go(\alpha_1,\ldots,\alpha_M) = \So^{\lambda\va_1}(\alpha_1) \, ,
\end{align*}
where 
\begin{align*}
\lambda = \frac{1}{M} \left( 1+ 2 \sum_{k=1}^{(M-1)/2 } \cos(\frac{k\pi}{M}) \right) \, .
\end{align*}
By Lagrange's trigonometric identity we have 
\begin{align*}
 \sum_{k=1}^{(M-1)/2 } \cos(\frac{k\pi}{M})  = -\half + \frac{1}{2 \sin(\frac{\pi}{2M})}\, .
\end{align*} 

\emph{Suppose then that $M$ is even}.
We choose 
\begin{align*}
& \vb_k= \cos(\theta_k +\frac{\pi}{2M}) \vx + \sin(\theta_k +\frac{\pi}{2M}) \vy \, , \\ 
& \theta_k=\frac{(k-1) \pi}{M}\, ,
\end{align*}
for $k=1,\ldots,M$. 
The first marginal of $\Go$ is
\begin{align*}
\sum_{\alpha_2,\ldots,\alpha_M=\pm 1} \Go(\alpha_1,\ldots,\alpha_M) = \So^{\lambda\va_1}(\alpha_1) \, ,
\end{align*}
where 
\begin{align*}
\lambda &= \frac{2}{M}  \sum_{k=1}^{M/2 } \cos(\frac{(2k-1)\pi}{2M})=\frac{1}{M\sin(\frac{\pi}{2M})} \, .
\end{align*}
\end{proof}

\subsection{Non-planar directions}\label{Platsec}

In this section we apply our method for qubit observables whose Bloch vectors are not in the same plane but are situated in a symmetric way. 
We explain the method for the case of three and four observables and we note that the noise parameter for these settings is the same. Moreover, we give the results concerning two more complicated cases, and the proofs for these results are given in the appendix. 

\subsubsection{Three observables}\label{sec:three-observables}

Let us first choose $\va_1 = \vx$, $\va_2 = \vy$ and $\va_3=\vz$. 
Adding also opposite directions the vectors would form an octahedron. 
However, since two vectors $\va$ and $-\va$ determine the same binary observable up to the permutation of outcomes, we keep only the positive directions.

For our adaptive strategy we choose $\vb_ 1= \frac{1}{\sqrt{3}}(\vx + \vy + \vz)$, $\vb_ 2= \frac{1}{\sqrt{3}}(-\vx + \vy + \vz)$, 
$\vb_ 3= \frac{1}{\sqrt{3}}(\vx - \vy + \vz)$, $\vb_ 4= \frac{1}{\sqrt{3}}(-\vx - \vy + \vz)$, see Fig.~\ref{fig:cube}. If in this case we measure for example in the direction of $\vb_1$ and get a positive (resp. negative) outcome our adaptive strategy assigns the value $+$ (resp. $-$) to all of the observables $\A_1,\A_2,$ and $\A_3$. The nonzero elements of the constructed joint observable $\Go$ are thus (see Fig.~\ref{fig:cube})
\begin{align*}
& \Go(+,+,+) = \frac{1}{4} \So^{\vb_1}(+) \, , & \Go(-,-,-) = \frac{1}{4} \So^{\vb_1}(-) \, ,  \\
& \Go(+,+,-) = \frac{1}{4} \So^{\vb_4}(-)  \, , & \Go(-,-,+) = \frac{1}{4} \So^{\vb_4}(+) \, , \\
& \Go(+,-,+) = \frac{1}{4} \So^{\vb_3}(+) \, ,  & \Go(-,+,-) = \frac{1}{4} \So^{\vb_3}(-) \, , \\
& \Go(+,-,-) = \frac{1}{4} \So^{\vb_2}(-)  \, , & \Go(-,+,+) = \frac{1}{4} \So^{\vb_2}(+) \, .
\end{align*}
The marginals of $\Go$ are $\So^{\lambda \va_1}$, $\So^{\lambda \va_2}$ and $\So^{\lambda \va_3}$ with the noise parameter $\lambda=1/\sqrt{3}\approx 0.5774$, and we have thus reproduce the result first proved in \cite{Busch86}.
Let us note that this noise parameter is known to be the boundary point, meaning that for any $\lambda>1/\sqrt{3}$ the three observables are incompatible \cite{Yu13}.

\begin{figure}[h!]
  \includegraphics{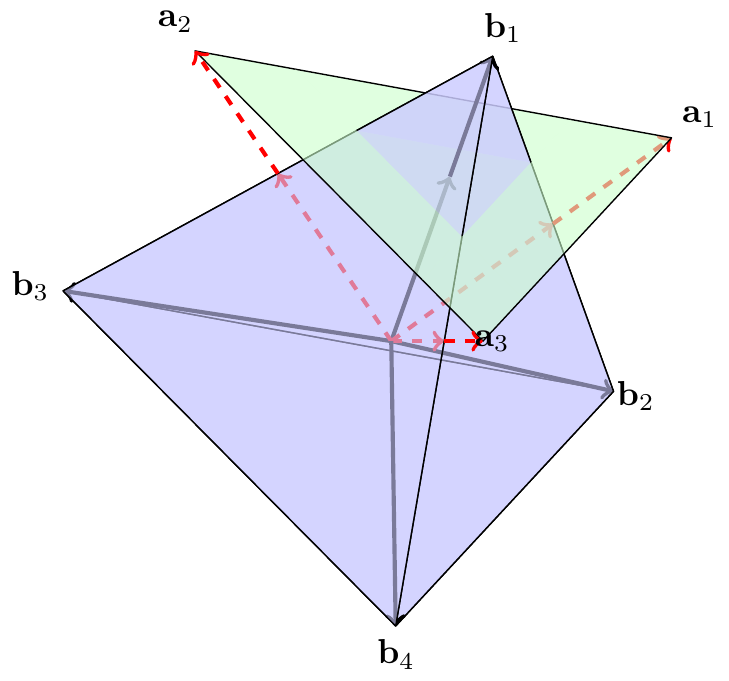}
  \caption{\label{fig:cube} (Color online) Three observables. The vectors $\va_i$, $i=1,2,3$, are the desired directions and the vectors $\vb_j$, $j=1,2,3,4$, are the guessing directions.}
\end{figure}
\subsubsection{Four observables}\label{sec:four-observables}
We then choose $\va_1 =  \frac{1}{\sqrt{3}}(\vx + \vy + \vz)$, $\va_2 = \frac{1}{\sqrt{3}}(\vx - \vy - \vz)$, 
$\va_3=\frac{1}{\sqrt{3}}(-\vx + \vy - \vz)$ and $\va_4=\frac{1}{\sqrt{3}}(-\vx - \vy + \vz)$, so the vectors go to the vertices of a tetrahedron.
For our adaptive strategy we choose $\vb_ 1=\vx$, $\vb_2=\vy$ and $\vb_3=\vz$, see Fig.~\ref{fig:tetra}. Now the nonzero elements of the joint observable $\Go$ are
\begin{align*}
& \Go(+,+,-,-) = \frac{1}{3} \So^{\vb_1}(+) \, , & \Go(-,-,+,+) = \frac{1}{3} \So^{\vb_1}(-) \, ,  \\
& \Go(+,-,+,-) = \frac{1}{3} \So^{\vb_2}(+) \, , & \Go(-,+,-,+) = \frac{1}{3} \So^{\vb_2}(-) \, , \\
& \Go(+,-,-,+) = \frac{1}{3} \So^{\vb_3}(+) \, , & \Go(-,+,+,-) = \frac{1}{3} \So^{\vb_3}(-) \, .
\end{align*}
The marginals of $\Go$ are $\So^{\lambda \va_j}$ with $\lambda=\frac{1}{\sqrt{3}}$. This value of the noise parameter coincides with the one for three observables (see previous subsection). 

\begin{figure}[h!]
  \includegraphics[width=0.4\textwidth]{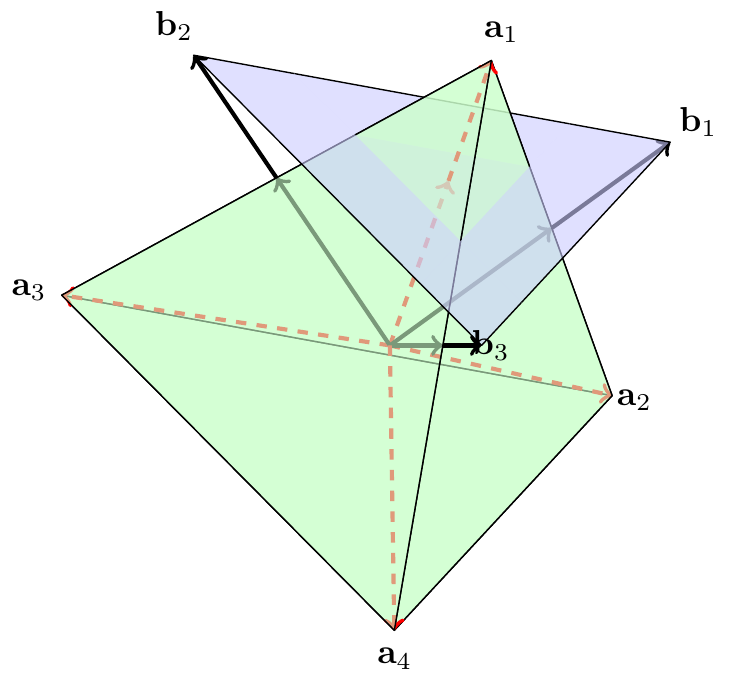}
  \caption{\label{fig:tetra} (Color online) Four observables. The vectors $\va_i$, $i=1,2,3,4$, are the desired directions and the vectors $\vb_j$, $j=1,2,3$, are the guessing directions.}
\end{figure}

\subsubsection{Six and ten observables}\label{sec:six-observables}

Icosahedron has 12 vertices and these directions determine 6 binary observables, while
dodecahedron has 20 vertices and these directions determine 10 binary observables.
Using the adaptive strategy we obtain the following values for the noise parameters (see Appendix): $\lambda_6=\frac{1+\sqrt 5}{6}\approx 0.5393$ and $\lambda_{10}=\frac{3+\sqrt{5}}{10}\approx 0.5236$. 
As we will later see, these are the least noise values making the binary observables jointly measurable.
 
\subsection{Adaptive strategy for the non-symmetric case}

The adaptive strategy gives an optimal joint observable also for the case of two arbitrary unbiased qubit observables.
For this purpose, let $\va_1$ and $\va_2$ be two Bloch vectors. We choose $\vb_ 1= \frac{1}{\|\va_1+\va_2\|}(\va_1 + \va_2)$ and $\vb_ 2= \frac{1}{\|\va_1-\va_2\|}(\va_1 - \va_2)$ to be the equal superpositions of these vectors.
The adaptive strategy gives a joint observable $\Go$ defined by
\begin{align*}
& \Go(+,+) = \mu \So^{\vb_1}(+)\, , \\
& \Go(+,-) = (1-\mu) \So^{\vb_2}(+)\, , \\
& \Go(-,+) = (1-\mu) \So^{\vb_2}(-)\, , \\
& \Go(-,-) = \mu \So^{\vb_1}(-),
\end{align*}
where the constant $\mu$ is the probability of making the measurement in the direction $\vb_1$ and $1-\mu$ is the probability of making the measurement in the direction $\vb_2$.

In order to get the correct marginals for $\Go$, i.e., $\So^{\lambda \va_1}$ and $\So^{\lambda \va_2}$, one needs to have 
\begin{align*}
\mu=\frac{\|\va_1+\va_2\|}{\|\va_1+\va_2\|+\|\va_1-\va_2\|}\, , 
\end{align*}
 and this gives
 \begin{align*}
 \lambda=\frac{2}{\|\va_1+\va_2\|+\|\va_1-\va_2\|} \, ,
 \end{align*}
which coincides with the optimal value originally presented in \cite{Busch86}.

\section{Generalised adaptive strategy for MUBs}\label{MUBadapt}

In order to use our strategy for measurements with more than two outcomes in a Hilbert space whose dimension is larger than or equal to two, one needs to introduce two minor modifications. First, in step three of the strategy a measurement of a two-valued observable $\So^{\vb_i}(\pm)$ is performed. In the general case such a measurement is not enough to distinguish between the outcomes of the desired observables. For example in the case where one of the desired observables (say $\A$) is three valued, a two valued observable could at best distinguish one of the outcomes of $\A$ and leave the other two open. This lack of distinguishability can be circumvented. Namely, instead of flipping a coin between two-valued observables, we build up the joint observable from the fixed vectors and not from the two-valued observables assigned to them.

Second, in the qubit case our strategy requires fixed Bloch vectors. The values of the joint observable are decided by checking the overlaps between these vectors and the Bloch vectors of the desired observables. In order to generalize our strategy we note that the inner products of the Bloch vectors (say $\va$ and $\vb$) and the inner products of the corresponding Hilbert space vectors (say $\psi_\va$ and $\psi_\vb$) are related by the formula
\begin{equation}
|\langle\psi_\va|\psi_\vb\rangle|^2 = \frac{1}{2}(1+\va\cdot\vb).
\end{equation}
Hence we can fix Hilbert space vectors instead of Bloch vectors and decide the values of the joint observable by maximising the inner products between the fixed rank-1 operators and the effects of the desired observales. More generally, one could fix higher rank effects and maximize the overlaps between these effects and the effects of the observables.

The generalisation withdraws the step three from our strategy. Step three can be seen as a requirement for the normalisation of the joint observable and, hence, we replace the third step by the condition that the fixed Hilbert space vectors as rank-1 operators must sum up to the identity.

To illustrate the generalised adaptive strategy we consider the case of two MUBs in a $d$-dimensional Hilbert space. Let $\{\varphi_j\}_{j=1}^n$ and $\{\psi_k\}_{k=1}^n$ be two mutually unbiased bases of $\mathbb C^n$, i.e., $\mo{\ip{\varphi_j}{\psi_k}}= 1/\sqrt{d}$ and let $\A_1(j)=|\varphi_j\rangle\langle\varphi_j|$ and $\A_2(k)=|\psi_k\rangle\langle\psi_k|$, $j,k=1,...,d$, be the corresponding observables.

We define unit vectors $b_{j,k}$ like in the previous section as equal superpositions of the desired directions $\varphi_j$ and $\psi_k$: 
\begin{align*}
b_{j,k}=N(\varphi_j +e^{i\theta_{j,k}}\psi_k)\, , \quad e^{i\theta_{j,k}}=\sqrt d\langle\psi_k|\varphi_j\rangle \, , 
\end{align*}
where $N$ is a normalization factor. It is easy to see that the overlaps $|\langle\varphi_m|b_{j,k}\rangle|^2$ and $|\langle\psi_n|b_{j,k}\rangle|^2$ are both maximal when $m=j$ and $n=k$. Hence we build up a joint measurement candidate by defining
\begin{align*}
\Go(j,k)=\frac{N^2}{d}|b_{j,k}\rangle\langle b_{j,k}|.
\end{align*}

It is now straight-forward to check that $\Go$ sums up to identity and that it gives as marginals the smeared versions of $\A_1$ and $\A_2$ with the amount
\begin{align*}
\lambda=\frac{1}{2}(1+\frac{1}{1+\sqrt d})\, ,
\end{align*}
of white noise.
This value has been earlier shown to be the optimal noise parameter in the case of Fourier connected MUBs \cite{CaHeTo12},\cite{Haapasalo15}.

\section{Proving optimality of adaptive joint measurements using steering}

In this latter part of our investigation we prove that the joint measurements considered in the previous sections are actually optimal. In other words, 
we show that their marginals possess the least possible amount of white noise, meaning that for any larger value of the noise parameter $\lambda$ the respective observables are incompatible. 
For proving the optimality, we employ some known steering techniques.
In this section we recall the basic setup of steering and explain why it helps to prove the optimality of the presented joint measurements.

Consider a bipartite scenario (Alice and Bob) sharing a quantum state $\rho_{AB}$. In a quantum steering task 
Alice tries to convince Bob that the shared state is entangled by making only local 
measurements $\A_k,\ k=1,...,n$ on her system and sending Bob the respective outcomes $x=1,...,m$ by classical communication. 

When Alice measures an observable $\A_k$ and gets an outcome $x$ Bob's conditional (non-normalized) state reads
\begin{equation}
\sigma_{x|k}=\text{tr}_A[(\A_k(x)\otimes \id)\rho_{AB}].
\end{equation}
We first notice that for any separable state $\sum_ip_i\rho_A^i\otimes\rho_B^i$ the conditional states always read
\begin{equation}
\sigma_{x|k}=\sum_i\text{tr}[\A_k(x)\rho_A^i]p_i\rho_B^i.
\end{equation}
This kind of an ensemble could also be created by classically post-processing a local set of (non-normalized) states 
$\{p_i\rho_B^i\}_i$ on Bob's side.

If, however, Bob runs over all local ensembles of non-normalized states $\{\sigma_i\}_i$ together with all possible post-processings and finds out that he can not reproduce his conditional states, he concludes that the shared state 
must be entangled. To be more precise, a bipartite setup (i.e. $\rho_{AB}$ and Alice's 
measurements $\{\A_k\}_{k=1}^n$) is non-steerable if there exists an ensemble of positive semi-definite operators $\{\sigma_\eta\}_\eta$ together with classical 
post-processings $p(x|k,\eta)$ such that
\begin{equation}\label{Steeringdef}
\sigma_{x|k}=\sum_\eta p(x|k,\eta)\sigma_\eta.
\end{equation}
If this is not the case, the setup is called steerable. For a steerable (resp. non-steerable) setup it is said that Alice can (resp. can not) steer Bob.

One notices that the definitions of steerability (\ref{Steeringdef}) and joint measurability (\ref{JMDef}) look very similar. It was proven in \cite{RU1, MQ} that this similarity originates from a one-to-one correspondence between these concepts: compatible measurements never allow Alice 
to steer Bob and steering is always possible with incompatible observables provided that the Schmidt rank of the shared pure state is $d$. In order to emphasize the link between steering and joint measurements it is convenient to use the maximally entangled state $|\psi\rangle=\frac{1}{\sqrt d}\sum_i |ii\rangle$ in a steering scenario. For the maximally entangled state the conditional states read
\begin{equation}
\sigma_{x|k}=\frac{1}{d}\A_k(x)^T\, ,
\end{equation}
where $T$ denotes the transpose. In what follows, we use this link to show that adaptive joint measurements are optimal in the sense that they possess the least possible amount of white noise.

\subsection{Necessary condition for qubit observables}
In this part we use known steering results to obtain the necessary condition for 
joint measurability of qubit observables. The steering inequality introduced in \cite{JoWi11} reads
\begin{align}\label{eq:4}
  \frac{1}{n}\sum_{k=1}^n \tr{(\A_k\otimes\vc_k\cdot\vsigma_k)\rho_{AB}} \leq C_n,
\end{align}
where $\A_k=\A_k(+)-\A_k(-)$, $\vc_k$ is a Bloch vector and $\rho_{AB}$ is a state of the composite system. The bound $C_n$ is the maximum value for the expression
\begin{align}
\frac{1}{n}\sum_{k=1}^n\sum_{x_k}x_k\tr{\sigma_{x|k}\vc_k\cdot\vsigma_k}\, ,
\end{align}
provided that the operators $\sigma_{x|k}$ are of the form (\ref{Steeringdef}). Here we have labeled the outcomes of the measurement $\A_k$ by $x_k$. The bound $C_n$ is obtained as 
\begin{align}\label{eq:5}
  C_n=\max_{x_k}\left(\lambda_{\rm max}\left(\frac{1}{n}\sum_{k=1}^n x_k\vc_k\cdot\vsigma_k\right)\right),
\end{align}
where $\lambda_{\rm max}\left(K\right)$ is the largest eigenvalue of a matrix $K$. 

For example, in the case of planar observables (Proposition \ref{sec:observables-plane}) one gets
\begin{align}
C_M^{\rm planar} = \frac{1}{M\sin\left(\frac{\pi}{2M}\right)}.
\end{align}
This is the maximum value for $C_M^{\rm planar}$ such 
that the scenario is not steerable. This means, like we discussed in the previous section, that a violation of the inequality (\ref{eq:4}) is only possible if Alice's observables are incompatible.

To see that in the planar case our adaptive strategy gives the optimal joint observable, let $\kb{\psi}{\psi}$ denote the maximally entangled state and notice that
\begin{align}\label{eq:6}
  \ptra{\A_k(x)\otimes \id \kb{\psi^\lambda}{\psi^\lambda}}=&\ptra{\A_k^\lambda(x)\otimes \id\kb{\psi}{\psi}},
\end{align}
where $\kb{\psi^\lambda}{\psi^\lambda} = \lambda \kb{\psi}{\psi} +\frac{1-\lambda}{4}\id\otimes\id$ and 
$\A_k^\lambda(x)=\lambda \A_k(x) + \frac{1-\lambda}{2}\tr{\A_k(x)}{}\mathbb{I}$. 

Inserting the noisy maximally entangled state $\kb{\psi^\lambda}{\psi^\lambda}$, the Bloch vectors of the planar observables from Proposition \ref{sec:observables-plane} for Bob, and the transposed planar observables for Alice\footnote{We choose transposed observables because using the properties of the maximally entangled state gives then a simple condition for steering:
\begin{align*}
\tr{\A_k^T\otimes\A_k\kb{\psi^\lambda}{\psi^\lambda}}=\lambda.
\end{align*}
Note that as this is the case for any choice of observables $\A_k$, one can use the equation (\ref{eq:5}) to calculate necessary joint measurability conditions for arbitrary qubit observables.}
to the left hand side of equation (\ref{eq:4}) one arrives with the condition $\lambda\leq C_M^{\rm planar}$. In other words, if $\lambda$ is larger than this threshold one violates the steering inequality and consequently Alice's observables are non-jointly measurable (transposition does not affect joint measurability) with the amount $\lambda$ of white noise (see equation (\ref{eq:6}). This means that joint measurability of Alice's observables implies that $\lambda\leq C_M^{\rm planar}$. Hence, we arrive with the following result

\begin{proposition}\label{prop:planar}
The planar observables introduced in Prop. \ref{sec:observables-plane} are jointly measurable if and only if 
$$
\lambda\leq\frac{1}{M\sin\left(\frac{\pi}{2M}\right)} \, .
$$
\end{proposition}

Calculating the values $C_M$ for the cases of Platonic solids and repeating the same procedure as above one finds that the joint measurements given by the adaptive strategy are optimal also in these cases. We summarize the results of this section in the following statement:

\begin{proposition}
The sufficient joint measurability conditions given in Section \ref{Planarsec} and \ref{Platsec} are also necessary.
\end{proposition}

\subsection{Necessary condition for MUBs}

In this part we prove that the joint measurability condition given in Section \ref{MUBadapt} are optimal.

For this purpose, suppose that Alice performs two different $d$-valued projective measurements given by two sets of MUBs $\{\varphi_j\}_{j=1}^d$ and $\{\psi_k\}_{k=1}^d$ as follows
\begin{equation}
\A_1(j)=\ket{\varphi_j}\bra{\varphi_j}, \:\: \A_2(k)=\ket{\psi_k}\bra{\psi_k}.
\end{equation}
The ensemble of states that Bob observes is denoted by $\{ \sigma_{x|i} \}_{x \in \mathbbm{Z}_d}$, $i=1,2$.

For this scenario the underlying state is non-steerable if and only if there exists a collection $\{ \omega_{jk}\}_{jk}$ of positive semidefinite operators $\omega_{jk} \geq 0$ such that
\begin{equation}
\label{eq:lhs}
\sigma_{j|1} = \sum_k \omega_{jk}, \:\: \sigma_{k|2} = \sum_j \omega_{jk},
\end{equation}
hold for all possible $j,k$. Note that here the notation $\omega_{jk}$ combines the post-processings and the positive operators in Eq.~(\ref{Steeringdef}). To certify that a decomposition of this form is impossible one needs to violate a steering inequality. For our purpose a useful inequality is characterized by a collection of operators $\{ Z_{jk} \}_{jk}$ which are all positive semidefinite $Z_{jk} \geq 0$ and furthermore satisfy the linear constraints 
\begin{equation}
\label{eq:LinCons}
Z_{jk} = Z_{jt} + Z_{sk} - Z_{st},
\end{equation}
for all possible $j,k,s,t$. If the observed ensembles $\{ \sigma_{j|1} \}_{j \in \mathbbm{Z}_d}$ and $\{ \sigma_{k|2}\}_{k \in \mathbbm{Z}_d}$ do have a decomposition of the above form, then it holds that
\begin{align}
\nonumber
\sum_{jk} &\tr{Z_{jk} \omega_{jk}} =\sum_{jk} \tr{(Z_{jt} + Z_{sk} - Z_{st}) \omega_{jk}} \\ 
\nonumber
&= \sum_j \tr{Z_{jt}(\sum_k \omega_{jk})} +\sum_k \tr{Z_{sk}(\sum_j \omega_{jk})}\\
&\ \ \ \ \ - \tr{Z_{st}(\sum_{j,k} \omega_{jk})} \\ 
\label{eq:FormSteerInequ}
& = \sum_j \tr{Z_{jt} \sigma_{j|1}} + \sum_k \tr{Z_{sk} \sigma_{k|2}} - \tr{Z_{st} \rho} \geq 0\, ,
\end{align}
for all $s,t$. Note that non-negativity holds because all $Z_{jk}$ and $\omega_{jk}$ are positive semidefinite. In addition we have used the fact that $\rho = \sum_j \sigma_{j|1} = \sum_k \sigma_{k|2}$. Hence if one violates this inequality one proves that such a positive semidefinite decomposition is impossible.

\subsection{Employed inequality}

In the following we use the following steering inequality. For the two given mutually unbiased basis sets $\{ \ket{\varphi_j} \}_{j \in \mathbbm{Z}_d}$, $\{ \ket{\psi_k} \}_{k \in \mathbbm{Z}_d}$ define 
\begin{equation}
\label{eq:Z1}
Z_{jk} = a \left( \ket{\varphi_j}\bra{\varphi_j} + \ket{\psi_k}\bra{\psi_k}\right) + b \id\, ,
\end{equation}
with 
\begin{equation}
a = -\frac{1}{(\sqrt{d}-1)(\sqrt{d}+2)}, \:\: b = \frac{\sqrt{d}+1}{\sqrt{d}(\sqrt{d}-1)(\sqrt{d}+2)}.
\end{equation}
In this form one directly verifies that all linear constraints given by Eq.~\eqref{eq:LinCons} are satisfied. In order to show that these operators are also positive semidefinite we express them in a different way.

Since the sets are mutually unbiased one has $\ip{\varphi_j}{\psi_k} = e^{i\theta_{jk}} /\sqrt{d}$. If one incorporates this phase into the superposition states
\begin{equation}
\ket{\chi^{\pm}_{jk}} = \frac{1}{N_{\pm}} \left( \ket{\varphi_j} \pm e^{-i \theta_{jk}} \ket{\psi_k} \right),
\end{equation}
one achieves that the resulting states become orthogonal $\ip{\chi^+_{jk}}{\chi^-_{jk}} = 0$. 
Furthermore note that the normalization is then also independent of $j$ and $k$.

Via these definitions one can rewrite the operators given by Eq.\eqref{eq:Z1} as 
\begin{equation}
\label{eq:Z2}
Z_{jk} = c \ket{\chi_{jk}^-}\bra{\chi_{jk}^-} + b(\id - \ket{\chi_{jk}^-}\bra{\chi_{jk}^-} - 
\ket{\chi_{jk}^+}\bra{\chi_{jk}^+})\, ,
\end{equation}
with 
\begin{equation}
c = \frac{2}{\sqrt{d}(\sqrt{d}-1)(\sqrt{d}+2)}.
\end{equation}
Since both terms in Eq.~\eqref{eq:Z2} consist of a positive constant multiplied with a positive semidefinite operator, this shows that all $Z_{kl}$ are positive semidefinite. 

Note that in order to derive the coefficients, one can start with $Z_{kl}$ as given by Eq.~\eqref{eq:Z2} but with open $b,c$. In order to achieve a form like Eq.~\eqref{eq:Z1} the coefficients in front of the terms $\ket{\varphi_j}\bra{\psi_k}$ and $\ket{\psi_j}\bra{\varphi_k} $ must vanish, which holds if
\begin{equation}
\frac{c-b}{N_-^2} + \frac{b}{N_+^2} = 0.
\end{equation}
The solution which additionally satisfies $\tr{Z_{kl}}=1$ is the given solution.

\subsection{The implication}
In this part we apply the steering inequality to derive the statement that the two observables $\A_1^\lambda$ and $\A_2^\lambda$ are not jointly measurable if $\lambda_{\rm max} \leq \lambda$.

We measure the maximally entangled state $\ket{\psi} = \frac{1}{\sqrt{d}}\sum_i \ket{ii}$ with the observables ${\A_1^{\lambda}}^T$ and ${\A_2^{\lambda}}^T$~\footnote{If we violate a steering inequality it means that ${\A_1^{\lambda}}^T$, ${\A_2^{\lambda}}^T$ are not jointly measureable; this implies further that also the transposed measurements are not jointly measurable.}. By the properties of the maximally entangled state, the conditional states are given by
\begin{equation}
\sigma_{j|1} = \frac{1}{d}\A_1^\lambda (j), \:\: \sigma_{k|2} = \frac{1}{d}\A_2^\lambda (k).
\end{equation}
If we now apply the steering inequality from Eq.\eqref{eq:FormSteerInequ} with the operators $Z_{kl}$ given by Eq.~\eqref{eq:Z1} one obtains
\begin{align*}
d \sum_{jk} & \tr{Z_{jk} \omega_{jk}} = \sum_j \tr{Z_{jt} \A_1^{\lambda}(j)} + \\
& + \sum_k \tr{Z_{sk} \A_2^{\lambda}(k)}  - \tr{Z_{st} \id} \\
&= 2d \left\{ \lambda \left[a \left(1 + \frac{1}{d}\right) +b \right] + \frac{1-\lambda}{d} \right\} - 1 \\&= 2 \left\{ 1- \lambda \left[ 1-a(d+1) -db\right]\right\} - 1 \\
&= 2 \left\{ 1-\lambda \left(1-\frac{1}{\sqrt{d}+2} \right)\right\} - 1 \geq 0.
\end{align*}
The inequality only holds if 
\begin{equation}
\lambda \leq \frac{1}{2}(1+\frac{1}{1+\sqrt d}) = \lambda_{\rm max},
\end{equation}
and this proves the statement.

\section{4-Specker set of qubit observables}

As a consequence of our previous results we can now show that there exists a so called 4-Specker set \cite{LiSpWi11} consisting of qubit observables. This means that there exists a set of four incompatible qubit observables which are not jointly measurable but any triplet of these observables forms a jointly measurable set. 
The precise set of these four observables is the following:

\begin{proposition}
Consider four qubit observables $\So^{\lambda\va_k}$ defined by equally distributed Bloch vectors on the upper side of the $xy$-plane, i.e. the angle between the vector $k$ and $k+1$ is $\pi/4$. There exists a smearing parameter $\lambda$ such that the observables $\So^{\lambda\va_k}$ form a 4-Specker set.
\end{proposition}

\begin{proof}
Because of the symmetry of the situation, every subset of three observables $\{\So^{\lambda\va_k}\}_{k\neq i}$ for some $i$ is jointly measurable if and only if (see formula 16 in \cite{Yu13})
\begin{align}
\lambda\leq(\cos(\pi/4)+2\sin(\pi/8))^{-1}\approx 0.679.
\end{align}
By Prop. \ref{prop:planar}, however, the set of four observables is incompatible if and only if $\lambda>\frac{1}{4\sin(\pi/8)}\approx 0.65$. 
Hence, choosing any $\lambda$ which is between these thresholds gives a 4-Specker set $\{\So^{\lambda\va_k}\}_{k=1}^4$ of qubit observables.
\end{proof}

It has been earlier shown that qubit observables can form a $3$-Specker set \cite{HeReSt08,LiSpWi11}, 
and that a $4$-Specker set exists in dimension $4$ \cite{KuHeFr14}.

\section{Conclusion}

Our method opens up possibilities for future research on both quantum incompatibility and steering. First, we have tested our method on basic symmetric and non-symmetric measurement setups and shown its power in these scenarios by reproducing some known and various new noise thresholds. The open question is how far can the method be fetched, i.e., how to find optimal vectors $\vb_i$ for more complicated measurement settings. Second, it is an open problem to decide which pairs of quantum observables are the most noise resistant ones in finite dimension. Our method gives lower bounds on the noise robustness of sets of measurements. Third, our technique might have interesting applications in quantum steering as every joint observable works as a local hidden state model for steering attempts with a restricted set of measurements \cite{RU2}. These questions are left for now for future works.

\begin{acknowledgments}
RU acknowledges the Finnish Cultural Foundation for financial support.

\end{acknowledgments}
\bibliographystyle{apsrev}

\begin{thebibliography}{39}

\bibitem{HeMiZi16}
T. Heinosaari, T. Miyadera, M. Ziman, J. Phys. A: Math. Theor. \ {\bf 49}, 123001, (2016).

\bibitem{Gühne04}
O. G\"uhne, Phys. Rev. Lett. \textbf{92}, 117903 (2004)

\bibitem{Busch14}
P. Busch, P. Lahti, R. Werner, Rev. Mod. Phys. \textbf{86}, 1261 (2014) 

\bibitem{Busch13}
P. Busch, P. Lahti, R. Werner, Phys. Rev. Lett. \textbf{111}, 160405 (2013)

\bibitem{WoPeFe09}
M.M.~Wolf, D.~Perez-Garcia, C.~Fernandez, Phys.~Rev.~Lett. \textbf{103}, 230402 (2009).

\bibitem{Brunner14}
N. Brunner,  D. Cavalcanti,  S. Pironio,  V. Scarani, S. Wehner, Rev. Mod. Phys. \textbf{86}, 419 (2014)

\bibitem{Heinosaari15}
T. Heinosaari, J. Kiukas, D. Reitzner, J. Schultz, J. Phys. A: Math. Theor. \textbf{48}, 435301 (2015)

\bibitem{MQ}
M. Quintino, T. Vertesi, and N. Brunner, Phys. Rev. Lett.\ {\bf 113}, 160402 (2014).

\bibitem{RU1}
R. Uola, T. Moroder, and O. G\"uhne, Phys. Rev. Lett.\ {\bf 113}, 160403 (2014).

\bibitem{Heinosaari152}
T. Heinosaari, J. Kiukas, D. Reitzner, Phys. Rev. A \textbf{92}, 022115 (2015)

\bibitem{Pusey15}
M. F. Pusey, J. Opt. Soc. Am. B 32, A56 (2015)

\bibitem{RU2}
R. Uola, C. Budroni, O. G\"uhne, and J.P. Pellonp\"a\"a, Phys. Rev. Lett.\ {\bf 115}, 230402 (2015).

\bibitem{Busch86}
P. Busch, Phys. Rev. D, 33, 2253--2261 (1986).

\bibitem{BuHe08}
P. Busch and T. Heinosaari, Quant. Inf. Comp., 8, 0797--0818, (2008).

\bibitem{Yu13}
S. Yu, C. H. Oh, arXiv:1312.6470

\bibitem{LiSpWi11}
Y.-C. Liang, R.W. Spekkens, H.M. Wiseman, Phys. Rep., \ {\bf 506}, 1--39, (2011).

\bibitem{CaHeTo12}
C. Carmeli, T. Heinosaari, and A. Toigo, Phys. Rev. A. \ {\bf 85}, 012109, (2012).

\bibitem{Haapasalo15}
E. Haapasalo, J. Phys. A: Math. Theor. \ {\bf 48}, 255303, (2015).

\bibitem{KuHeFr14}
R. Kunjwal, C. Heunen, and T. Fritz, Phys. Rev. A. \ {\bf 89}, 052126, (2014).

\bibitem{Heinosaari13}
T. Heinosaari, J. Phys. A: Math. Theor. \textbf{46}, 152002, (2013).

\bibitem{Zhu15}
H. Zhu, Sci. Rep. \textbf{5}, 14317, (2015).


\bibitem{JoWi11}
S. J. Jones, H. M. Wiseman, Phys. Rev. A \textbf{84}, 012110 (2011)

\bibitem{Saunders10}
D. J. Saunders, S. J. Jones, H. M. Wiseman, G. J. Pryde, Nature Physics \textbf{6}, 845–849 (2010)

\bibitem{HeReSt08}
T. Heinosaari, D. Reitzner and P. Stano, Found. Phys., \ {\bf 38}, 1133--1147, (2008).

\end{thebibliography}

\onecolumngrid

\clearpage

\appendix

\section{6 observables}

Denote $\chi=\frac{1+\sqrt{5}}{2}$ and $n=\frac{5+\sqrt{5}}{2}$.
Let 
\begin{align*}
\va_ 1= \frac{1}{\sqrt{n}}(\vy + \chi\vz) \, , \quad & \va_ 2= \frac{1}{\sqrt{n}}(-\vy + \chi\vz)\, , \\
\va_ 3= \frac{1}{\sqrt{n}}(\vx + \chi\vy) \, , \quad & \va_ 4= \frac{1}{\sqrt{n}}(-\vx + \chi\vy)\, , \\
\va_ 5= \frac{1}{\sqrt{n}}(\chi\vx + \vz) \, , \quad & \va_ 6= \frac{1}{\sqrt{n}}(-\chi\vx + \vz).
\end{align*}
We choose $\vb_i=\va_i$ for $i=1,\ldots,6$ and obtain the elements of the joint observable $\Go$ according to the 
Fig.~\ref{fig:ico} as follows:

\begin{figure}[h!]
  \includegraphics[width=0.4\textwidth]{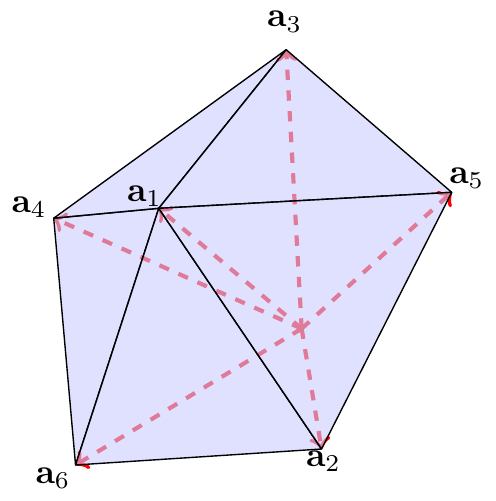}
  \caption{\label{fig:ico} (Color online) Six observables (icosahedron). The desired directions $\va_i$ are the same as the 
    guessing directions $\vb_j$.}
\end{figure}

The nonzero elements of the constructed joint observable $\Go$ are
  \begin{align*}
    & \Go(+,+,+,+,+,+) = \frac{1}{6} \So^{\vb_1}(+)\, ,
    & \Go(-,-,-,-,-,-) = \frac{1}{6} \So^{\vb_1}(-)\, ,\\
    & \Go(+,+,-,-,+,+) = \frac{1}{6} \So^{\vb_2}(+)\, ,
    & \Go(-,-,+,+,-,-) = \frac{1}{6} \So^{\vb_2}(-)\, ,\\
    & \Go(+,-,+,+,+,-) = \frac{1}{6} \So^{\vb_3}(+)\, ,
    & \Go(-,+,-,-,-,+) = \frac{1}{6} \So^{\vb_3}(-)\, ,\\
    & \Go(+,-,+,+,-,+) = \frac{1}{6} \So^{\vb_4}(+)\, ,
    & \Go(-,+,-,-,+,-) = \frac{1}{6} \So^{\vb_4}(-)\, ,\\
    & \Go(+,+,+,-,+,-) = \frac{1}{6} \So^{\vb_5}(+)\, ,
    & \Go(-,-,-,+,-,+) = \frac{1}{6} \So^{\vb_5}(-)\, ,\\
    & \Go(+,+,-,+,-,+) = \frac{1}{6} \So^{\vb_6}(+)\, ,
    & \Go(-,-,+,-,+,-) = \frac{1}{6} \So^{\vb_6}(-).
  \end{align*}
Marginals of $\Go$ are $\So^{\lambda \va_i}$, $i\in 1,\ldots,6$  with $\lambda=\frac{1+\sqrt{5}}{6}$.

\section{10 observables}

We denote $\chi = \frac{1+\sqrt{5}}{2}$.
Let
\begin{align*}
\va_1= \frac{1}{\sqrt{3}}(\vx+\vy+\vz) \, , \quad & \va_2= \frac{1}{\sqrt{3}}(\vx-\vy+\vz)\, , \\
\va_3= \frac{1}{\sqrt{3}}(-\vx + \vy + \vz) \, , \quad & \va_ 4= \frac{1}{\sqrt{3}}(\vx +\vy-\vz)\, , \\
\va_5= \frac{1}{\sqrt{3}}(\chi^{-1}\vy+\chi\vz)\,,\quad &\va_6=\frac{1}{\sqrt{3}}(-\chi^{-1}\vy+\chi\vz)\, ,\\
\va_7= \frac{1}{\sqrt{3}}(\chi^{-1}\vx+\chi\vy)\,,\quad &\va_8=\frac{1}{\sqrt{3}}(-\chi^{-1}\vx+\chi\vy)\, ,\\
\va_9= \frac{1}{\sqrt{3}}(\chi\vx+\chi^{-1}\vz)\,,\quad &\va_{10}=\frac{1}{\sqrt{3}}(-\chi\vx+\chi^{-1}\vz).
\end{align*}
We choose $\vb_i=\va_i$ for $i=1,\ldots,10$ and obtain the outcomes of the joint observable $\Go$ according to the following table:

\begin{center}
\begin{tabular}{| c | c | c |}
\hline
$k$ & outcome of $\So^{\vb_k}$  &  outcome of $\Go$ \\
\hline\hline
1 & + & (+,+,+,+,+,+,+,+,+,-)  \\
 & - &  (-,-,-,-,-,-,-,-,-,+) \\
 \hline
2 & + & (+,+,-,-,+,+,-,-,+,-)\\
  & - & (-,-,+,+,-,-,+,+,-,+) \\
 \hline
3 & + & (+,-,+,-,+,+,+,+,-,+) \\
  & - & (-,+,-,+,-,-,-,-,+,-) \\
\hline
4 & + & (+,-,-,+,-,-,+,+,+,-) \\
  & - & (-,+,+,-,+,+,-,-,-,+) \\
\hline
5 & + & (+,+,+,-,+,+,+,+,+,+) \\
  & - & (-,-,-,+,-,-,-,-,-,-) \\
\hline
6 & + & (+,+,+,-,+,+,-,-,+,+) \\
  & - & (-,-,-,+,-,-,+,+,-,-) \\
\hline
7 & + & (+,-,+,+,+,-,+,+,+,-) \\
  & - & (-,+,-,-,-,+,-,-,-,+) \\
\hline
8 & + & (+,-,+,+,+,-,+,+,-,+) \\
  & - & (-,+,-,-,-,+,-,-,+,-) \\
\hline
9 & + & (+,+,-,+,+,+,+,-,+,-) \\
  & - & (-,-,+,-,-,-,-,+,-,+) \\
\hline
10 & + &(-,-,+,-,+,+,-,+,-,+) \\
  & - & (+,+,-,+,-,-,+,-,+,-) \\
\hline
\end{tabular}
\label{tab:dodecahedron}
\end{center}

The nonzero elements of the constructed joint observable $\Go$ are the ones given by the table \ref{tab:dodecahedron} divided by the number of guessing observables. Marginals of $\Go$ are $\So^{\lambda \va_i}$, $i\in 1,\ldots,10$  with $\lambda=\frac{1+\chi}{5}$.

\end{document}